\documentclass[runningheads,11pt]{llncs}
\usepackage[T1]{fontenc}
\usepackage{graphicx}
\usepackage[margin=1in]{geometry}
\usepackage{hyperref}
\usepackage{color}

\usepackage{amsmath,amsfonts,amssymb,bm,bbm}
\newcommand{\hz}{\hat{z}}

\newcommand{\ba}{\bm{a}}
\newcommand{\bb}{\bm{b}}
\newcommand{\be}{\bm{e}}
\newcommand{\bw}{\bm{w}}
\newcommand{\bx}{\bm{x}}
\newcommand{\by}{\bm{y}}
\newcommand{\bz}{\bm{z}}

\newcommand{\ox}{\overline{x}}
\newcommand{\os}{\overline{s}}

\newcommand{\bzero}{\bm{0}}

\newcommand{\bR}{\mathbb{R}}
\newcommand{\bZ}{\mathbb{Z}}

\DeclareMathOperator*{\argmax}{arg\,max}

\begin{document}
\title{Continuous-Time Best-Response and Related Dynamics
\\ in Tullock Contests with Convex Costs}
\titlerunning{Continuous-Time Best-Response Dynamics in Tullock Contests}

%
\author{Edith Elkind \and
Abheek Ghosh \and
Paul W.~Goldberg}
\authorrunning{E. Elkind, A. Ghosh, and P.W. Goldberg}
\institute{Department of Computer Science, University of Oxford, OX1 3QG, U.K.\\
\email{\{edith.elkind,abheek.ghosh,paul.goldberg\}@cs.ox.ac.uk}}
\maketitle

\begin{abstract}
% \sloppy
Tullock contests model real-life scenarios that range from competition among proof-of-work blockchain miners to rent-seeking and lobbying activities. We show that continuous-time best-response dynamics in Tullock contests with convex costs converges to the unique equilibrium using Lyapunov-style arguments. We then use this result to provide an algorithm for computing an approximate equilibrium. We also establish convergence of related discrete-time dynamics, e.g., fictitious play.
These results indicate that the equilibrium is a reliable predictor of the agents' behavior in these games.
\end{abstract}

\section{Introduction}
Contests are games where agents compete by making costly investments to win valuable prizes. The model of Tullock~\cite{tullock1980efficient} is one of the most widely used for studying these environments, and has been applied to scenarios that originate from economics, political science, and computer science.
Applications include cryptocurrencies (proof-of-stake~\cite{bahrani2024centralization}, proof-of-work~\cite{leshno2020bitcoin}), rent-seeking, political lobbying, competition among firms, crowdsourcing, resource allocation and routing, college admissions, and sports~(see, e.g., the books \cite{konrad2009strategy,vojnovic2015contest} surveying some of these applications). 
As a concrete example, consider the game among miners in proof-of-work cryptocurrencies such as Bitcoin~\cite{chen2019axiomatic,leshno2020bitcoin}. A Bitcoin miner's expected reward for the next block is proportional to her costly computational effort. A Tullock contest closely approximates the game among these miners: a set of agents compete to win a valuable prize by investing costly effort, and an agent gets a portion of the prize proportional to his effort. 
In general, Tullock contests can model a large class of competitive environments where the reward of an agent increases with her own effort, decreases with others' effort, and the cost increases with her own effort.

Given an environment with strategic agents, such as a contest, 
it is desirable to be able to reliably predict the agents' behavior, so as to reason about possible outcomes. 
% Also, when applicable, one may want to design the environment to elicit certain behaviors. 
Nash equilibrium strategies serve as an initial approximation to this goal, but the existence of an equilibrium is not always a good predictor of the agents' behavior. Indeed, the traditional explanation of equilibrium is that it results from analysis and introspection by the agents in a situation where the rules of the game, the rationality of the agents, and the agents’ payoff functions are all common knowledge. Both conceptually and empirically, these assumptions are not always satisfied in real-life scenarios~\cite{fudenberg1998theory}. A model provides a more robust prediction of the outcome of a game if it explains how the outcome can be attained in a decentralized manner, ideally via a process that involves agents responding to incentives provided
by their environment.
Variations of best-response (BR) dynamics are arguably the simplest and the most intuitive of these models. In BR dynamics, agents sequentially change their current strategy towards one that best-responds to that of the other agents. This framework is especially well-suited to settings such as Tullock contests with convex costs, where pure-strategy equilibria are guaranteed to exist.

Recent papers (e.g., \cite{ewerhart2017lottery,ghosh2023best2}) have explored discrete-time BR dynamics in Tullock contests with convex costs. In this model, at each time step, an agent updates their current action to their BR action. These papers show that this dynamics converges for homogeneous agents, i.e., when all agents have the same cost function. On the other hand, when the agents are non-homogeneous, the dynamics does not converge, and this non-convergence result holds for many instances.\footnote{The set of instances (cost functions, starting state of the dynamics) on which the dynamics does not converge has a positive measure with respect to the set of all possible instances.}
Intuitively, these results imply that we can expect simple myopic agents to reach the equilibrium if they have similar cost functions, but not necessarily if they have different cost functions.

In real-life situations, agents may not be able to instantly change their current actions to their BR actions. They are more likely to slowly move from their current actions to actions with better utility. For example, in the game among Bitcoin miners, a miner may gradually buy more equipment and thereby increase their computational power if they see an opportunity to improve their utility by doing so. The change may be more rapid if the discrepancy is substantial, i.e., if the distance between the current action and the optimal BR action is larger, but we may still expect the change to be smooth enough that other agents can react before a given agent moves from their current action to their BR action. This observation also holds for other applications of Tullock contests: e.g., in a competition among firms for research and development of drugs and vaccines~\cite{dasgupta1986theory}, a firm may generally be able to only slowly change their output by hiring or firing researchers or increasing or decreasing their investment.

The continuous-time BR dynamics~\cite{fudenberg1998theory} is the classic model for studying such smoothly changing actions and is the main focus of this paper. We show that continuous-time BR dynamics converges to the unique equilibrium when agents have arbitrary, possibly non-homogeneous, weakly convex cost functions.\footnote{The primary focus in the literature has been on convex cost functions, because they model increasing marginal cost per unit utility, or, equivalently, decreasing marginal utility per unit cost, which is a feature of many economic environments. Also, Tullock contests with non-convex cost functions generally do not have a pure-strategy Nash equilibrium (PNE), and BR dynamics are not stable, by definition, when there does not exist a PNE. Hence, such models need to be studied using alternative equilibrium concepts and learning dynamics that accommodate mixed strategies.} 
Our rate-of-convergence bound is tight: the dynamics converges to an $\epsilon$-approximate equilibrium in $\Theta(\log(1/\epsilon))$ time. We prove this result using a Lyapunov potential function that measures the total regret, as perceived by the agents, for playing their current action instead of their BR action. 
We then extend this analysis to show convergence in certain classes of discrete-time dynamics, e.g., when the agents take small steps towards the BR (where the step size is not necessarily limiting to $0$ as assumed for continuous time) or when the agents play fictitious-play-like dynamics (best-respond to the empirical average action of other agents).
These results also lend to a simple algorithm to compute an $\epsilon$-approximate Nash equilibrium in time polynomial in $1/\epsilon$ and parameters of the model.

Overall, our results indicate that several learning dynamics that are only slightly more involved than the standard discrete-time BR dynamics circumvent the non-convergence result for non-homogeneous agents. So, we should expect even non-homogeneous agents with convex cost functions to converge to the unique equilibrium in Tullock contests.

\section{Related Work}\label{sec:related}
\cite{ewerhart2017lottery,ghosh2023best1,ghosh2023best2} study discrete-time BR dynamics in Tullock contests. 
\cite{ewerhart2017lottery,ghosh2023best1} study a special case of Tullock contests where the cost functions are linear, also known as lottery contests.
\cite{ewerhart2017lottery} shows that a lottery contest with homogeneous agents is a BR potential game~\cite{voorneveld2000best,kukushkin2004best}; the class of BR potential games is a strict generalization of the better-known classes of ordinal and exact potential games~\cite{monderer1996potential}.
Lottery contests are not exact potential games even with homogeneous agents~\cite{ewerhart2017lottery}, and they are not ordinal potential games with non-homogeneous agents~\cite{ewerhart2020ordinal}.

\cite{ghosh2023best1} uses the BR potential function of \cite{ewerhart2017lottery} to derive bounds on the rate of convergence of BR dynamics in lottery contests with homogeneous agents, assuming one random agent moves at a time.
It uses techniques from convex optimization (coordinate descent~\cite{wright2015coordinate}) and randomized algorithms (Markov chains).
\cite{ghosh2023best2} extends the convergence results of \cite{ghosh2023best1} for lottery contests to Tullock contests with convex costs (for homogeneous agents) by introducing an adversarial dynamic process called the discounted-sum dynamics.

In contrast to the papers above, we study continuous-time BR dynamics and show convergence for arbitrary non-homogeneous convex cost functions. In our analysis, we use a natural potential function that measures the total regret as perceived by the agents for playing their current action instead of their BR action, and use Lyapunov arguments to prove convergence. This potential function has been used before in the different context of proving convergence of stochastic fictitious play in two-player zero-sum games~\cite{hofbauer2002global,shamma2004unified}; however, these results do not extend to zero-sum games with three or more players.\footnote{Any $n$-player game can be written as an $(n+1)$-player zero-sum game.} Our analysis works because of the special structure exhibited by Tullock contests.

A Tullock contest is an aggregative game (games where agent $i$'s utility depends upon her action and the sum of the actions of all agents). \cite{thorlund1990iterative} is related to Section~\ref{sec:fictitious} of our paper. It shows that the fictitious play converges in aggregative games if the BR function's slope has an absolute value less than $1$. They use a contraction argument (intuitively, the next state can be written as a linear operation of the current state, where the eigenvalues of the linear operator have a magnitude of strictly less than $1$). The BR function's slope in Tullock contests can be unbounded, and the techniques do not carry over. 

In a related paper~\cite{dindovs2006better}, the authors investigate a randomized discrete better-response dynamics in aggregative games. In their model, a random agent selects a better-response action at each time step based on a probability distribution. Crucially, this distribution must assign positive probabilities to every subset of better-response actions of positive measure. This ensures that eventually the dynamics will get sufficiently close to the equilibrium, and using additional properties of the game near the equilibrium, the dynamics converges. While this paper shares a similar motivation with our work, the actual dynamics being studied differ significantly, and the techniques employed do not directly apply. The paper also uses a continuous $1$-dimensional dynamic process in their analysis, but the continuous BR dynamics we study is an $n$-dimensional process, which requires different techniques.

A Tullock contest also corresponds to a Cournot oligopoly with isoelastic inverse demand. There are convergence results for learning dynamics that apply to specific types of Cournot games, such as Cournot oligopoly with strictly declining BR functions~\cite{deschamps1975algorithm,thorlund1990iterative}, Cournot game with linear demand~\cite{slade1994does}, aggregative games that allow monotone BR selections~\cite{huang2002fictitious,dubey2006strategic,jensen2010aggregative}, and others~\cite{dragone2012static,bourles2017altruism}. However, all these methods do not apply to the Tullock contest whose BR function is not monotone~\cite{dixit1987strategic}.

A different line of research studies the convergence of learning dynamics in supermodular~\cite{milgrom1990rationalizability,vives1990nash} and submodular~\cite{topkis1979equilibrium} games. Tullock contests are neither supermodular nor submodular, nor generalizations of these concepts in the literature~\cite{milgrom1994monotone,amir1996cournot}, and the results and techniques in the related literature do not carry over.

\cite{moulin1978strategically} implicitly shows the strategic equivalence between contests and zero-sum games. This directly implies convergence of fictitious play dynamics in Tullock contests with two agents over a discretized action space~\cite{ewerhart2020fictitious}, but no such result has been proven for three or more agents. 
A different line of research has studied the convergence (or chaotic behavior) of learning dynamics in other types of contests (such as all-pay auctions, e.g.,~\cite{puu1991chaos,warneryd2018chaotic,cheung2021learning}), but these techniques and results also do not apply to Tullock contests.
\section{Preliminaries}\label{sec:prelim} 
\sloppy
In Tullock's model~\cite{tullock1980efficient}, $n$ agents participate in a contest with unit prize (normalized). Let $[n] = \{1, 2, \ldots, n\}$. 
The agents simultaneously produce non-negative output; we denote the output of agent $i$ by $x_i \in \bR_{\ge 0}$ and the output profile by $\bx = (x_1, \ldots, x_n) \in \bR_{\ge 0}^n$. Let $\bx_{-i} = (x_1, \ldots, x_{i-1}, x_{i+1}, \ldots, x_n)$ and $s_{-i} = \sum_{j \neq i} x_j$.
Agent $i$ incurs a cost of $c_i(x_i)$ 
%EE: you make a bunch of implicit assumptions on c, like it being twice differentiable; do you want to state them here?
%AA: Stated in the next paragraph.
for producing the output $x_i$ and receives a fraction of the prize proportional to $x_i$ 
if at least one agent produces a strictly positive output, and $1/n$ otherwise.\footnote{Some papers in the literature, 
e.g., Dasgupta and Nti~\cite{dasgupta1998designing}, assume that if $s=0$, all agents get a prize of $0$. 
Our analysis and results remain the same under this alternate assumption as well.} The utility of agent $i$, $u_i(\bx)$, is
\begin{equation}\label{eq:utility:single}
    u_i(\bx) = \frac{x_i}{\sum_j x_j} - c_i(x_i) = \frac{x_i}{x_i + s_{-i}} - c_i(x_i).
\end{equation}
Notice that the utility of agent $i$ depends upon her output $x_i$ and the total output of other agents $s_{-i}$, but not upon the distribution of $s_{-i}$ across the $n-1$ agents. To emphasize this, we will write $u_i(\bx)$ as $u_i(x_i, s_{-i})$.

% We make the following standard assumptions on the cost functions: for every agent $i$, (a) zero cost for non-participation, $c_i(0) = 0$; (b) increasing, $c_i'(z) > 0$ for all $z > 0$; (c) weakly convex, $c_i''(z) \ge 0$ for all $z \ge 0$. Implicit in these assumptions is also the twice-differentiability of $c_i$.
% Without these assumptions, a pure-strategy Nash equilibrium may not exist (see, e.g., \citet[Chapter 4]{vojnovic2015contest}, \citet{ewerhart2020unique}, and references therein).

We make the following assumptions on the cost functions: for every agent $i$, the cost function $c_i$ is (a) twice differentiable; (b) zero cost for non-participation, $c_i(0) = 0$; (c) increasing, $c_i'(z) > 0$ for all $z > 0$; (d) weakly convex, $c_i''(z) \ge 0$ for all $z \ge 0$.
These assumptions are standard in the literature, and without them, a pure-strategy Nash equilibrium may not exist (see, e.g., \cite[Chapter 4]{vojnovic2015contest}, \cite{ewerhart2020unique}). 

\begin{remark}\label{rk:logit}
Our model is equivalent to contests of general logit form with concave success functions and convex costs~\cite[Chapter 4]{vojnovic2015contest}. In logit contests, the utility of agent $i$ is given by 
% \[
%     \widehat{u}_i(\widehat{\bx}) = \frac{\widehat{f_i}(\widehat{x_i})}{\sum_j \widehat{f_j}(\widehat{x_j})} - \widehat{c_i}(\widehat{x_i}),
% \]
$\widehat{u}_i(\widehat{\bx}) = \frac{\widehat{f_i}(\widehat{x_i})}{\sum_j \widehat{f_j}(\widehat{x_j})} - \widehat{c_i}(\widehat{x_i})$, 
where $\widehat{x_j}$ is the output of agent $j$. For each $i$, if $\widehat{f_i}$ is concave and $\widehat{c_i}$ is convex (and both are non-negative and strictly increasing), then we can do the change of variables $x_i = \widehat{f_i}(\widehat{x_i})$ and set $c_i(x_i) = \widehat{c_i}(\widehat{f_i}^{-1}(x_i))$ to write the utility function as given in equation \eqref{eq:utility:single}.
% $u_i(\bx)  = \frac{x_i}{\sum_j x_j} - c_i(x_i) = \widehat{u_i}(\widehat{\bx})$. 
Similarly, a utility function of the form $\frac{V_i x_i}{\sum_j x_j} - c_i(x_i)$, where $V_i$ is the value of the prize for agent $i$, can be converted to the form given in \eqref{eq:utility:single} by scaling down the utility function by $V_i$, which does not affect the strategy of the agent.
Also note that Tullock contests with utility functions of the form $\widehat{u}_i(\widehat{\bx}) = \frac{\widehat{x_i}^r}{\sum_j \widehat{x_j}^{r}} - \widehat{c_i}(\widehat{x_i})$ for some $r \in (0,1]$ are special cases of logit contests where $\widehat{f_i}(\widehat{x_i}) = \widehat{x_i}^r$.
\end{remark}

\subsection{Best-Response}\label{sec:prelim:br}
Given the total output $s_{-i} = \sum_{j \neq i} x_j$ of all agents except $i$, the best-response (BR) of agent $i$ is an action $x_i$ such that
\begin{equation*}
    x_i \in \argmax_{z \ge 0} u_i(z, s_{-i}) = \argmax_{z \ge 0} \left( \frac{z}{z + s_{-i}} - c_i(z) \right) .
\end{equation*}
First, note that an agent has no BR if the output of every other agent is $0$.\footnote{
Indeed, if $s_{-i} = 0$, then by producing an output of $\epsilon > 0$, agent $i$ gets a utility of $u_i(\epsilon, 0) = 1 - c_i(\epsilon)$. Since $c_i$ is continuous and $c_i(0) = 0$, 
for sufficiently small $\epsilon$ we have $c_i(\epsilon)<1-1/n$ and hence $u_i(\epsilon, 0)> u_i(0, 0) = 1/n$. Thus, $x_i = 0$ cannot be a BR. Further, no $\epsilon > 0$ can be a BR, because $c_i$ is increasing and thus
$u_i(\epsilon/2, 0) = 1 - c_i(\epsilon/2) > 1 - c_i(\epsilon) = u_i(\epsilon, 0)$.
}
To circumvent this technical issue, we make the exogenous assumption that agent $i$ plays some small positive action if $s_{-i} = 0$.\footnote{This issue of not having any BR to $0$ can also be resolved by an alternate technical assumption: the prize is given to agent $i$ with probability $\frac{x_i}{b + \sum_j x_j}$ and agent $i$'s expected utility is $\frac{x_i}{b + \sum_j x_j} - c_i(x_i)$, where $b$ is a small positive constant. Under this alternate assumption the prize may remain unallocated with a positive probability of $\frac{b}{b + \sum_j x_j}$, unlike in our model. We \textit{expect} all results in this paper to hold for this alternate model as well.\label{fn:alternateModel}}
We allow this action to be arbitrary (and can possibly change over time) as long as it is positive and is at most $\frac{1}{2\max_i c_i'(0)}$. Moreover, as we will observe in our analysis, this case occurs only for a small period at the start of the dynamics and never afterward.
On the other hand, agent $i$ has a unique BR if $s_{-i} > 0$, i.e., if the output produced by at least one other agent $j$ is non-zero. This unique BR can be computed by taking a derivative of $u_i(z, s_{-i})$ with respect to $z$. If $s_{-i} > 0$, then 
\begin{align}
    \frac{\partial u_i(z, s_{-i})}{\partial z} &= \frac{ s_{-i} }{(z + s_{-i})^2} - c_i'(z), \label{eq:du} \\
    \frac{\partial^2 u_i(z, s_{-i})}{\partial z^2} &= \frac{-2 s_{-i} }{(z + s_{-i})^3} - c_i''(z) < 0, \label{eq:ddu}
\end{align}
where the last inequality holds because $z \ge 0$, $c_i''(z) \ge 0$, and $s_{-i} > 0$. So, $u_i(z, s_{-i})$ is strictly concave w.r.t.~$z$, i.e., $\frac{\partial u_i(z, s_{-i})}{\partial z}$ is strictly decreasing w.r.t.~$z$. 
Let $BR_i(s_{-i})$ denote the BR of agent $i$ given that the total output of other agents is $s_{-i}$. The first-order conditions for the BR are
\begin{align}\label{eq:br}
    BR_i(s_{-i}) > 0 &\text{ and } \left. \frac{\partial u_i(z, s_{-i})}{\partial z} \right|_{z = BR_i(s_{-i})} = 0,\qquad \text{ if } \left. \frac{\partial u_i(z, s_{-i})}{\partial z} \right|_{z = 0} > 0; \\
    BR_i(s_{-i}) = 0 &\text{ and } \left. \frac{\partial u_i(z, s_{-i})}{\partial z} \right|_{z = BR_i(s_{-i})} \le 0,\qquad \text{ if } \left. \frac{\partial u_i(z, s_{-i})}{\partial z} \right|_{z = 0} \le 0.
\end{align}
% If $s_{-i} > 0$, notice that at $z = \max(1, (c_i)^{-1}(1))$, 
%EE did you really mean c_i' above? or c_i?
% the agent has a cost of at least $1$ but a prize of less than $1$, therefore $ \left. \frac{\partial u_i(z, s_{-i})}{\partial z} \right|_{z = \max(1, c_i^{-1}(1))} < 0$. So, the BR is less than $\max(1, c_i^{-1}(1))$. For $s_{-i} = 0$, the BR is $a < 1$ by assumption.

\subsection{Continuous-Time Best-Response Dynamics}
Let $\bx(t) = (x_{i}(t))_{i \in [n]}$ denote the action profile of the agents at time $t$ in the BR dynamics. Similarly, let $s(t) = \sum_j x_{j}(t)$ and $s_{-i}(t) = \sum_{j \neq i} x_{j}(t)$. 
The continuous-time (or simply, continuous) BR dynamics starts from an initial profile $\bx(0) = (x_{i}(0))_{i \in [n]} \in \bR_{\ge 0}^n$. At time $t \ge 0$, each agent $i \in [n]$ continuously updates their action as
\begin{equation}\label{eq:dx}
    \frac{d x_i(t)}{d t} = BR_i(s_{-i}(t)) - x_{i}(t).
\end{equation}

The continuous BR dynamics is a limiting case of the discrete BR dynamics when the step size goes to $0$. In discrete BR, for a step size of $\Delta t = 1$, $x_{i}(t+1) = BR_{i}(s_{-i}(t)) \Longleftrightarrow x_{i}(t+1) - x_{i}(t) = BR_{i}(s_{-i}(t)) - x_{i}(t)$ for any given agent $i$. For arbitrary $\Delta t \ge 0$, this dynamics can be generalized to $x_{i}(t+\Delta t) - x_{i}(t) = \Delta t (BR_{i}(s_{-i}(t)) - x_{i}(t)) \Longleftrightarrow \frac{x_{i}(t+\Delta t) - x_{i}(t)}{\Delta t} = BR_{i}(s_{-i}(t)) - x_{i}(t)$. The continuous dynamics takes the limit $\Delta t \rightarrow 0$.

\subsection{Equilibrium}
A Tullock contest with weakly convex cost functions always has a pure-strategy Nash equilibrium (which is also the unique equilibrium, including mixed-strategy Nash equilibria, see, e.g., \cite{ewerhart2020unique}). So, we exclusively focus on pure equilibria in this paper. 
% Throughout the paper, we will denote the unique equilibrium of the Tullock contest by $\bx^*$.
\begin{definition}[Pure-Strategy Nash Equilibrium]
    An action profile $x^* = (x_1^*, \ldots, x_n^*)$ is a pure-strategy Nash equilibrium if it satisfies
    \[
        u_i (x_i^*, \bx_{-i}^*) \ge u_i (x_i', \bx_{-i}^*), 
    \]
    for every agent $i$ and every action $x_i'$ for agent $i$.
\end{definition}

In general, the BR dynamics in a Tullock contest may never \textit{exactly} reach the equilibrium in finite time, rather it may \textit{converge} to the equilibrium. 
The dynamics converges to an equilibrium if it reaches an $\epsilon$-approximate equilibrium in finite time for any $\epsilon > 0$.
\begin{definition}[Approximate Pure-Strategy Nash Equilibrium]
    An action profile $\bx = (x_1, \ldots, x_n)$ is an $\epsilon$-approximate pure-strategy Nash equilibrium, for $\epsilon > 0$, if it satisfies
    \[
        u_i (x_i, \bx_{-i}) \ge u_i (x_i', \bx_{-i}) - \epsilon, 
    \]
    for every agent $i$ and every action $x_i'$ for agent $i$.
\end{definition}

\section{Convergence of Continuous-Time Best-Response Dynamics}\label{sec:cont}
In this section, we prove that the continuous BR dynamics given in equation \eqref{eq:dx} rapidly converges to the unique pure-strategy Nash equilibrium, i.e., $\bx(t) \rightarrow \bx^*$ as $t \rightarrow \infty$, where $\bx^*$ is the equilibrium. 

\begin{theorem}\label{thm:contConverge}
The continuous best-response dynamics $\bx(t)$ in Tullock contests with weakly convex cost functions converges to an $\epsilon$-approximate pure-strategy Nash equilibrium in $O(\log(1/\epsilon))$ time. Further, there are instances where reaching an $\epsilon$-approximate equilibrium takes $\Omega(\log(1/\epsilon))$ time.
\end{theorem}

Note that linear dynamical systems converge in $\Theta(\log(1/\epsilon))$ time. In our dynamics, see equation~\eqref{eq:dx}, the $-x_i(t)$ term is linear but the $BR_i(s_{-i})$ term is non-linear, so a $\Theta(\log(1/\epsilon))$ convergence can be expected but is not obvious.
We present the proof for the upper bound stated in Theorem~\ref{thm:contConverge} below and the lower bound in Appendix~\ref{app:proofs}.

\begin{proof}[Theorem~\ref{thm:contConverge} (Upper Bound)]
To prove this convergence result, we use the potential function given in \eqref{eq:V}.
This potential has been used previously to prove the convergence of stochastic fictitious play for two-player zero-sum games with finite action space; see Section~\ref{sec:related} for further discussion.
For an action profile $\bx$,
% with at least two agents $i$ and $j \neq i$ with positive output $x_i > 0$ and $x_j > 0$, 
let the potential function $V(\bx)$ be defined as:
\begin{align}\label{eq:V}
    V(\bx) &= \sum_{i \in [n]} V_i(\bx), \\
    \text{ where } V_i(\bx) &= \max_{z} u_i(z, s_{-i}) - u_i(x_i, s_{-i}) = u_i(BR_i(s_{-i}), s_{-i}) - u_i(x_i, s_{-i}) . \nonumber
\end{align}
Notice that $\max_{z} u_i(z, s_{-i})$ is not well-defined if $s_{-i} = 0$. In this case, as discussed in Section~\ref{sec:prelim}, we assume that $BR_i(0) \in (0, \frac{1}{2 \max_i c_i'(0)}]$, and $V_i(\bx) = u_i(BR_i(0), 0) - u_i(x_i, 0)$. We will prove that such profiles can only occur during a short initial phase of the BR dynamics.

$V_i(\bx)$ measures agent $i$'s regret for playing $x_i$ instead of the best possible action given $s_{-i}$, i.e., it is the amount of utility that agent $i$ can increase by playing the BR instead of $x_i$. Notice that, by definition,
\begin{enumerate}
    \item $V(\bx) \ge 0$. Because $u_i(BR_i(s_{-i}), s_{-i}) \ge u_i(x_i, s_{-i}) \Longleftrightarrow V_i(\bx) \ge 0$ for every agent $i$ and profile $\bx$, which implies $V(\bx) \ge 0$.
    \item $V(\bx) = 0$ at the equilibrium. 
    $V(\bx) = 0 \Longleftrightarrow \bx = \bx^*$ because $ V(\bx) = 0 \Longleftrightarrow  V_i(\bx) = 0, \forall i \in [n] \Longleftrightarrow u_i(BR_i(s_{-i}), s_{-i}) = u_i(x_i, s_{-i}), \forall i \in [n]$.
\end{enumerate}
Given the profile $\bx(t)$, we can write the potential at time $t$ as $V(\bx(t)) = \sum_i V_i(\bx(t))$. For conciseness, we denote $V_i(\bx(t))$ by $V_i(t)$, or simply $V_i$, when the dependency on $\bx(t)$ is clear from the context; similarly, $V(\bx(t))$ by $V(t)$ or $V$. Given the dynamics followed by $\bx(t)$, equation \eqref{eq:dx}, we can write the dynamics that $V(t)$ follows as
\begin{equation}\label{eq:dV}
     \frac{d V}{dt} = \sum_{i \in [n]} \frac{\partial V}{\partial x_i} \frac{d x_i}{dt}.
\end{equation}

We next bound the time it takes to reach a state with two positive outputs, and we show that this property always holds thereafter.
% Let us bound the time it takes to reach such a state.

\paragraph{Warm-Up Phase}
First, notice that if $x_i(\tau) > 0$ at some time point $\tau$, then $x_i(t) > 0$ for all $t \ge \tau$ because 
% $x_i$ follows the dynamics given in equation \eqref{eq:dx}, so 
$\frac{dx_i(t)}{dt} = BR(s_{-i}(t)) - x_{i}(t) \ge - x_i(t) \implies \frac{d x_i(t)}{x_i(t)} \ge -dt \implies x_i(t) \ge x_i(\tau) e^{-(t - \tau)} > 0$. 
So, once we reach a state with two agents $i$ and $j \neq i$ with positive output, then these two agents will always have a positive output thereafter.

Say we start from a profile $\bx(0) = \bzero$, i.e., all agents have $0$ output initially. By our technical assumption discussed in Section~\ref{sec:prelim}, the action of an agent $i$ is some small constant in $(0, \frac{1}{2 \max_i c_i'(0)}]$, say $\eta_i$. So, $\frac{dx_i}{dt} = \eta_i > 0$ at time $t=0$ for all $i$, which implies that at time $dt > 0$ with $dt \rightarrow 0$, we have $x_i(dt) > 0$, as required.

Let us now consider the case when there is only one agent $i$ with $x_i(0) = \alpha > 0$, and all other agents $j \neq i$ have $x_j(0) = 0$. 
Now, if $x_i(0) = \alpha < \frac{1}{c_j'(0)}$ for some $j \neq i$, then $BR_j(s_{-j}(0)) = BR_j(\alpha) > 0$ because by the first-order condition, equation \eqref{eq:du}, for $z = 0$, we have
\begin{align*}
    \frac{\partial u_i(z, \alpha)}{\partial z} &= \frac{\alpha}{(z + \alpha)^2} - c_j'(z) = \frac{1}{\alpha} - c_j'(0) > 0.
\end{align*}
So, at time $0$, we have $\frac{dx_j}{dt} = BR_j(s_{-j}(0)) - x_j(0) = BR_j(\alpha) > 0$, which implies that at time $dt > 0$ with $dt \rightarrow 0$, we have $x_j(dt) > 0$. 

Now, let us consider the case when $x_i(0) = \alpha \ge \max_{j \neq i} \frac{1}{c_j'(0)}$. Let $\beta = \max_{j \neq i} \frac{1}{c_j'(0)}$ for conciseness. Let us bound the time---denoted by $T$---it takes to reach $x_i(T) < \beta$. As $x_i(t) \ge \beta$ for all $t < T$, we have $\frac{d x_i(t)}{dt} = BR_i(s_{-i}(t)) - x_i(t) = BR_i(0) - x_i(t) = \eta_i - x_i(t) \le \frac{\beta}{2} - x_i(t)$, where the last inequality holds because $\eta_i \le \max_j \frac{1}{2 c_j'(0)} \le \max_{j \neq i} \frac{1}{2 c_j'(0)} = \frac{\beta}{2}$. Using this, we get
\begin{equation*}
    \frac{d x_i(t)}{dt} \le \frac{\beta}{2} - x_i(t) \implies \frac{d x_i(t)}{x_i(t) - \frac{\beta}{2}} \le -dt \implies \ln \left( \frac{x_i(t) - \frac{\beta}{2}}{\alpha - \frac{\beta}{2}} \right) \le -t,
\end{equation*}
which implies that to reach $x_i(T) < B$, it is sufficient to have $T > \ln\left( \frac{\alpha - \frac{\beta}{2}}{\beta - \frac{\beta}{2}} \right) = \ln\left( \frac{2 \alpha}{\beta} - 1\right)$. Further notice that $V(\bx(0)) = V_i(\bx(0)) = (1 - c_i(\eta_i)) - (1 - c_i(x_i(0))) = x_i(0) - \eta_i \ge \alpha - \frac{\beta}{2}$. So, within $T = O(\log(V(0)))$ time, we get at least two agents with positive output, and this property holds thereafter.

\paragraph{Main Phase} Given our analysis of the warm-up phase above, from here on we assume that there are always two agents with positive output. We next prove the following lemma about $V(t)$.
\begin{lemma}\label{lm:V}
    The potential $V(t) = V(\bx(t))$ at any time $t$ satisfies the following differential inequality
    \[
        \frac{d V}{dt} + V \\
        \le -\sum_i \frac{y_i}{\sum_j y_j} \left( 1 - \frac{\sum_j y_j}{y_i + s_{-i}} \right)^2 \le 0,
    \]
    where the dependency on $t$ is suppressed,
    where $y_i(t) = BR_i(s_{-i}(t))$,
    % where $p_i(t) = \frac{BR_i(s_{-i}(t))}{\sum_j BR_j(s_{-j}(t))}$ and $q_i(t) = \frac{s_{-i}(t)}{\sum_j BR_j(s_{-j}(t))}$, 
    and assuming that there are at least two agents with positive output in the profile $\bx(t)$.
\end{lemma}
\begin{proof}%[Proof of Lemma~\ref{lm:V}]
Let us suppress the dependency on $t$, e.g., let us write $\bx(t)$ as $\bx$ and $V(t)$ as $V$. Let $y_i = BR_i(s_{-i})$ for conciseness.

Note that for general convex cost function $c_i$, we do not have a closed-form formula for $BR_i$. But from the first-order conditions, equation \eqref{eq:br}, we have
\begin{align*}
    y_i > 0 \text{ \& } \left. \frac{\partial u_i(z, s_{-i})}{\partial z} \right|_{z = y_i} = 0,&  \text{ if } \left. \frac{\partial u_i(z, s_{-i})}{\partial z} \right|_{z = 0} > 0; \\
    y_i = 0 \text{ \& } \left. \frac{\partial u_i(z, s_{-i})}{\partial z} \right|_{z = y_i} \le 0,&  
    \text{ if } \left. \frac{\partial u_i(z, s_{-i})}{\partial z} \right|_{z = 0} \le 0.
\end{align*}
Now, $ \left. \frac{\partial u_i(z, s_{-i})}{\partial z} \right. = \frac{s_{-i}}{(z + s_{-i})^2} - c_i'(z)$; plugging in $z = 0$ we get $\left. \frac{\partial u_i(z, s_{-i})}{\partial z} \right|_{z = 0} = \frac{s_{-i}}{(0 +s_{-i})^2} - c_i'(0) = \frac{1}{s_{-i}} - c_i'(0)$. The condition $\left. \frac{\partial u_i(z, s_{-i})}{\partial z} \right|_{z = 0} > 0$ corresponds to $\frac{1}{s_{-i}} - c_i'(0) > 0 \Longleftrightarrow s_{-i} c_i'(0) < 1$. Similarly, $\left. \frac{\partial u_i(z, s_{-i})}{\partial z} \right|_{z = 0} \le 0$ corresponds to $ s_{-i} c_i'(0) \ge 1$. Using these, the first-order conditions at $y_i = BR_i(s_{-i})$ can be rewritten as
\begin{align}
    \frac{s_{-i}}{(y_i + s_{-i})^2} = c_i'(y_i), & \text{ if $s_{-i} c_i'(0) < 1$,} \label{eq:br1} \\
    y_i = 0, & \text{ if $s_{-i} c_i'(0) \ge 1$}. \label{eq:br2}
\end{align}
We can write $V$ as
\begin{align}
    V &= \sum_i V_i = \sum_i( u_i(y_i, s_{-i}) - u_i(x_i, s_{-i}) )  \label{eq:V1} \\
    &= \sum_i \left( \frac{y_i}{y_i + s_{-i}} - c_i(y_i) - \frac{x_i}{x_i + s_{-i}} + c_i(x_i) \right) = \sum_i \frac{y_i}{y_i + s_{-i}} - \sum_i c_i(y_i) + \sum_i c_i(x_i) - 1. \nonumber
\end{align}
The time derivative of $V$ w.r.t.~$t$ is $\frac{d V}{d t} = \sum_k \frac{\partial V}{\partial x_k} \frac{d x_k}{d t}$, where $\frac{d x_k}{d t} = y_k - x_k$ in continuous BR dynamics. 
To write $\frac{\partial V}{\partial x_k}$, we need to know $\frac{\partial y_i}{\partial x_k}$ for all $i$ and $k$. Note that $\frac{\partial y_i}{\partial x_k} = 0$ for $k = i$ and $\frac{\partial y_i}{\partial x_k} = \frac{d y_i}{d s_{-i}}$ for $k \neq i$.  Due to the constraint that $y_i$ (the best-response) is always non-negative, there is a point of non-differentiability at $y_i = 0$. In particular, if $c_i'(0) > 0$, then:
\begin{itemize}
    \item If $s_{-i} > \frac{1}{c_i'(0)}$, then in the small neighborhood around $s_{-i}$, say $[s_{-i} - \delta, s_{-i} + \delta]$ for small $\delta > 0$, we will have the corresponding $y_i = 0$. So, $\frac{d y_i}{d s_{-i}} = 0$.
    \item If $s_{-i} < \frac{1}{c_i'(0)}$, then $y_i > 0$ and is governed by equation \eqref{eq:br1}. We can differentiate this equation w.r.t.~to $s_{-i}$ to get
    \begin{align}
        &\frac{1}{(y_i + s_{-i})^2} - \frac{2 s_{-i}}{(y_i + s_{-i})^3} = \left( c_i''(y_i) + \frac{2 s_{-i}}{(y_i + s_{-i})^3} \right) \frac{d y_i}{d s_{-i}} \nonumber \\
        \implies &\frac{d y_i}{d s_{-i}} = \frac{y_i - s_{-i}}{2 s_{-i} + (y_i + s_{-i})^3 c_i''(y_i)}. \label{eq:dBR1}
    \end{align}
    If we take the limit $s_{-i} \uparrow \frac{1}{c_i'(0)}$, which implies that $y_i \downarrow 0$, we get 
    \begin{align*}
        \frac{d y_i}{d s_{-i}} = \frac{- 1}{2 + c_i''(0)/(c_i'(0))^2} < 0,
    \end{align*}
    where the last inequality is true because $c_i'(0) > 0$ and $c_i''(0) \ge 0$. On the other hand, we know that the magnitude is bounded $\frac{d y_i}{d s_{-i}} = \frac{- 1}{2 + c_i''(0)/(c_i'(0))^2} \ge \frac{-1}{2}$.
\end{itemize}
The above two cases tell us that $\frac{d y_i}{d s_{-i}}$ at $s_{-i} = \frac{1}{c_i'(0)}$ has a left limit strictly less than $0$ but a right limit equal to $0$, so we have non-differentiability at $s_{-i} = \frac{1}{c_i'(0)}$. But as $\frac{d y_i}{d s_{-i}}$ is bounded near $s_{-i} = \frac{1}{c_i'(0)}$, we can define it to be equal to some finite value in $[\frac{- 1}{2 + c_i''(0)/(c_i'(0))^2}, 0]$ at $s_{-i} = \frac{1}{c_i'(0)}$, which is sufficient for our analysis.
An alternate analysis can be done using the \textit{envelop theorem} of Milgrom and Segal~\cite[Theorem 2]{milgrom2002envelope} to arrive at the same result.
% and use the \textit{envelop theorem} of \citet[Theorem 2]{milgrom2002envelope} to do the analysis, which is discussed in Appendix~\ref{app:nonDiffBR}. Below we present an analysis assuming $s_{-i} \neq \frac{1}{c_i'(0)}$ for all $i$, which captures the essential steps. 

Taking partial derivative of $V$ w.r.t.~$x_k$, we get
\begin{align*}
    \frac{\partial V}{\partial x_k} &= \frac{\partial }{\partial x_k} \left( \sum_i \frac{y_i}{y_i + s_{-i}} - \sum_i c_i(y_i) + \sum_i c_i(x_i) - 1 \right) \\
    &= \sum_i \left( \frac{s_{-i}}{(y_i + s_{-i})^2} - c_i'(y_i) \right) \frac{\partial y_i}{\partial x_k} - \sum_{i \neq k} \frac{y_i}{(y_i + s_{-i})^2} + c_k'(x_k).
\end{align*}
From the discussion above, we know that either $\frac{\partial y_i}{\partial x_k} = 0$ or $\frac{s_{-i}}{(y_i + s_{-i})^2} = c_i'(y_i)$ and $\frac{\partial y_i}{\partial x_k}$ is bounded. In either case, we have $\left( \frac{s_{-i}}{(y_i + s_{-i})^2} - c_i'(y_i) \right) \frac{\partial y_i}{\partial x_k} = 0$, so
\begin{align}
    \frac{\partial V}{\partial x_k} = c_k'(x_k) - \sum_{i \neq k} \frac{y_i}{(y_i + s_{-i})^2}. \label{eq:dV1}
\end{align}
Putting together, we can write $\frac{d V}{d t}$ as
\begin{align*}
    \frac{d V}{d t} &= \sum_k \frac{\partial V}{\partial x_k} \frac{d x_k}{d t} = \sum_k (y_k - x_k) c_k'(x_k) - \sum_k (y_k - x_k) \sum_{i \neq k} \frac{y_i}{(y_i + s_{-i})^2} \\
    &= \sum_i (y_i - x_i) c_i'(x_i) - \sum_i \frac{y_i  (\sigma - y_i - s_{-i})}{(y_i + s_{-i})^2},
\end{align*}
where $\sigma = \sum_k y_k$. Adding $V$ and $\frac{d V}{d t}$ together, we get
\begin{align*}
    V + \frac{d V}{d t} &= \sum_i \frac{y_i}{y_i + s_{-i}} - \sum_i c_i(y_i) + \sum_i c_i(x_i) - 1 + \sum_i (y_i - x_i) c_i'(x_i) - \sum_i \frac{y_i (\sigma - y_i - s_{-i})}{(y_i + s_{-i})^2} \\
    &= -1 + \sum_i (\underbrace{- c_i(y_i) + c_i(x_i) + (y_i - x_i) c_i'(x_i)}_{\le 0 \text{ as $c_i$ is convex}}) + \sum_i \frac{2 y_i}{y_i + s_{-i}} - \sum_i \frac{y_i \sigma}{(y_i + s_{-i})^2}.
\end{align*}
Now, let $p_i = \frac{y_i}{\sum_j y_i} = \frac{y_i}{\sigma}$ and $q_i = \frac{s_{-i}}{\sigma}$. Notice that $\sum_i p_i = 1$ and, for all $i$, $p_i \ge 0$ and $q_i \ge 0$. Plugging this into the inequality above, we get
\begin{align*}
    V + \frac{d V}{d t} &\le -1 + \sum_i \frac{2 y_i}{y_i + s_{-i}} - \sum_i \frac{y_i \sigma}{(y_i + s_{-i})^2} = -1 + \sum_i \frac{2 p_i}{p_i + q_i} - \sum_i \frac{p_i}{(p_i + q_i)^2} \\
    &= \sum_i p_i \left( -1 + \frac{2}{p_i + q_i} - \frac{1}{(p_i + q_i)^2} \right) = - \sum_i p_i \left( 1 - \frac{1}{p_i + q_i} \right)^2 \le 0.
\end{align*}
\qed \end{proof}
Let us now use Lemma~\ref{lm:V} to get the desired rate of convergence upper bound. We use standard Lyapunov arguments:
\begin{equation*}
    \frac{d V(t)}{dt} + V(t) \le 0 \implies \frac{d V(t)}{V(t)} \le -dt \implies V(t) \le V(0) e^{-t}.
\end{equation*}
For any $t \ge \ln\left( \frac{1}{\epsilon} \right) + \ln(V(0))$, we get $V(t) \le \epsilon$, which implies that for every agent $i$ we have $V_i(t) = V_i(\bx(t)) \le \epsilon \Longleftrightarrow u_i(\bx(t)) \ge u_i(BR_i(s_{-i}(t)), s_{-i}(t)) - \epsilon$. So, we are at an $\epsilon$-approximate equilibrium.
This completes the proof for the upper bound.
\qed \end{proof}

\section{Discrete-Time Dynamics (Approximately Continuous BR and Fictitious Play) and Equilibrium Computation}\label{sec:approx}
In this section, we consider discrete-time BR dynamics. We also provide an algorithm for computing an approximate equilibrium based on such dynamics. Proofs are given in Appendix~\ref{app:proofs}.

Let us consider a modification to the original Tullock contest model we have studied till now. We assume that each agent $i$ must always play an action $x_i \ge x_{\min}$ instead of $x_i \ge 0$, for some $x_{\min} \ge 0$. 
Notice that $x_{\min} = 0$ corresponds to the original model, while a $x_{\min} > 0$ says that any participant in the contest must have a positive minimum output.
An assumption of $x_{\min} > 0$ may be plausible in practical scenarios where there is a positive cost of participation (the cost of \textit{showing up} for the game).
% implies that every participant in the game must put at least some effort (and bear the associated cost) of 
% Even if $x_{\min}$ is assumed to be strictly positive if it is allowed to be small enough, it will not be an unreasonable assumption. I, which seems plausible for practical problems. 
We also normalize the cost functions and assume that $\min_i c_i(1) = 1$ for all $i$; this ensures that any rational agent will always play an action $\le 1$. We also assume that the second derivative and the ratio of the first derivatives of the cost functions are bounded: $\frac{\max_{i, z \in [x_{\min},1]} c'_i(z)}{\min_{i, z \in [x_{\min},1]} c'_i(z)} = B_1$ and $\max_{i, z \in [x_{\min},1]} c''_i(z) = B_2$. 
% We will assume that $x_{\min}$  $x_{\min} < 1/B_1$.\footnote{}

The first-order conditions for the case when $x_{\min} > 0$ are similar to the ones given in \eqref{eq:br} except that the critical point above which the first-order condition is satisfied with equality is $x_{\min}$ instead of $0$. 
% In particular, if $s_{-i}$
The analysis for the continuous BR dynamics for this model is also analogous to the analysis for Theorem~\ref{thm:contConverge}.

Let us now consider BR dynamics in this model with the step-size, say $\Delta t$, small but not necessarily going to $0$. In particular,
\begin{align}\label{eq:delx}
    x_{i}(t+\Delta t) = x_{i}(t) + \Delta t \cdot (BR_{i}(s_{-i}(t)) - x_{i}(t)).
\end{align}
The continuous-time BR dynamics corresponds to equation \eqref{eq:delx} with $\Delta t \rightarrow 0$. We aim to find bounds on $\Delta t$ that ensure convergence.

\begin{lemma}\label{lm:suf}
For a profile $\bx$, let $H(\bx)$ be defined as
\begin{align*}
    H(\bx) = \frac{ \frac{B_2}{2} \sum_i (y_i - x_i)^2 + \sum_{i \in E} \frac{(\sigma - y_i - s_{-i})^2}{s_{-i}^2} }{\sum_i \frac{y_i (\sigma - y_i - s_{-i})^2}{\sigma (y_i + s_{-i})^2} },
\end{align*}
where $y_i = BR_i(s_{-i})$ and $\sigma = \sum_i y_i$.
% where $p_i = \frac{BR_i(s_{-i})}{\sum_j BR_j(s_{-j})}$ and $q_i = \frac{s_{-i}}{\sum_j BR_j(s_{-j})}$. 
If the step-size at time $t$ is bounded above by $1/\max(2, H(\bx(t))$, then the BR dynamics converges to the unique equilibrium. In particular, for $0 < \alpha_t \le 1/\max(2, H(\bx(t))$, we have $V(t+\alpha_t) \le (1 - \alpha_t) V(t)$.
\end{lemma}
% \input{proofs/suf}
% \paragraph{$x_{\min} > 0$}
If $x_{\min}$ is assumed to be strictly positive, then we can upper bound $H(\bx)$ as a function of $x_{\min}$. 
\begin{theorem}\label{thm:suf1}
    If $x_{\min} > 0$, then $H(\bx) = O\left(\frac{n(1+B_2)}{x_{\min}^3}\right)$ for all $\bx$, which implies that the dynamics reaches an $\epsilon$-approximate equilibrium in $O\left(\frac{1}{\alpha} \log\left(\frac{V(0)}{\epsilon}\right)\right)$ steps with a suitable step-size $\alpha = \Theta\left(\frac{x_{\min}^3}{n(1+B_2)}\right)$.
\end{theorem}
Notice that the bound in Theorem~\ref{thm:suf1} depends upon the number of agents $n$. This is essential, as highlighted by Lemma~\ref{lm:nec-n} below.
\begin{lemma}\label{lm:nec-n}
% To guarantee convergence of the dynamics, it is necessary to have the step-size of $O(\frac{1}{n})$, even for linear and homogeneous cost functions.
If the step-size is not $O(1/n)$, then there are instances with linear and homogeneous cost functions where the dynamics does not converge.
\end{lemma}

Note that although the bound in Theorem~\ref{thm:suf1} does not depend upon $B_1$ (the ratio of the first-derivatives of the cost functions of the agents, which measures the relative skills of the agents), $B_1$ may be implicit in $x_{\min}$. If $x_{\min}$ is not sufficiently small, e.g., if $x_{\min} = \omega(1/B_1^2)$, then the equilibrium when the agents are restricted to play $x_i \ge x_{\min}$ may be different from the equilibrium when the agents can play less than $x_{\min}$. 
For example, for two agents with linear cost functions $c_1(x_1) = x_1$ and $c_2(x_2) = \beta x_2$, where $\beta \ge 1$, the unique equilibrium is $\bx^* = \left( \frac{\beta}{(1+\beta)^2}, \frac{1}{(1+\beta)^2} \right)$ if the agents can play any $x_i \ge 0$. Note that $B_1 = \beta$, so the equilibrium output of agent-$2$ is $\Theta(1/B_1^2)$. On the other hand, if the agents are restricted to play $x_i \ge x_{\min} = \omega(1/B_1^2)$, then the equilibrium will be forced to be different from $\bx^*$.
Given this observation, it would be natural to assume that $x_{\min}$ is small enough; in particular, $x_{\min} = O(1/B_1^2)$.
Moreover, the dependency on $B_1$ is unavoidable, as formalized by Lemma~\ref{lm:nec-cost} below.
\begin{lemma}\label{lm:nec-cost}
% To guarantee convergence of the dynamics, it is necessary to have the step-size of $O(1/B_1)$ for linear cost functions, even for two agents.
If the step-size is not $O(1/B_1)$, then there are instances with two agents and linear cost functions where the dynamics does not converge.
\end{lemma}

If $x_{\min} = 0$, then our results do not provide a lower bound on the step-size that is independent of the action profile $\bx(t)$. In particular, Lemma~\ref{lm:suf} does not directly imply such a bound because $H(\bx)$ may be unbounded for some profiles $\bx$, then the step-size recommended by the lemma at $\bx$ to ensure convergence goes to $0$. Indeed, such a lower bound might not exist. 
% Because, if it did exist, then it would imply that by simulating the BR dynamics we can compute an $\epsilon$-approximate equilibrium for a largely unrestricted class of convex cost functions in $O(\log(1/\epsilon))$ steps, which seems unlikely. 
On the other hand, even in the case of $x_{\min} = 0$, we can simulate with a pseudo $\hat{x}_{\min} = \Theta(\epsilon)$ to compute an equilibrium in $poly(1/\epsilon, n, B_1, B_2)$ steps as shown below.

\paragraph{Algorithm} Let us construct a modified game with a pseudo lower bound on the outputs of the agents: $\hat{x}_{\min} = \epsilon/(4B_1)$. We simulate the BR dynamics in this game with a step-size of $\alpha = \Theta\left(\frac{\hat{x}_{\min}^3}{n(1+B_2)}\right)$, as recommended by Theorem~\ref{thm:suf1}, to compute an $(\epsilon/2)$-approximate equilibrium of this modified game in $O\left(\frac{1}{\alpha} \log\left(\frac{V(0)}{\epsilon}\right)\right)$ steps.
Let this approximate equilibrium be $\hat{\bx}$. 
At $\hat{\bx}$, all agents have a regret of at most $\epsilon/2$ assuming that they can only play above $\hat{x}_{\min}$. By playing below $\hat{x}_{\min}$, they can further increase their utility by at most $c_i(\hat{x}_{\min}) - c_i(0) \le B_1 \hat{x}_{\min} / (1 - \hat{x}_{\min}) \le 2 B_1 \hat{x}_{\min} \le \epsilon/2$. So, at $\hat{\bx}$, the total regret of any agent in the original game is at most $\epsilon$, as required.

\subsection{Fictitious Play: Best-Response to Empirical Average}\label{sec:fictitious}
Let us consider a discrete-time dynamics with a step-size of $\Delta t = 1$, but where the agents best-respond to the empirical average action of the other agents. 
Let $\ox_i(t) = \frac{1}{t} \sum_{\tau = 1}^t x_i(t)$ and $\os_{-i}(t) = \sum_{j \neq i} \ox_j(t)$.
Formally, the dynamics is defined as follows: the action of agent $i$ at time $t+1$ is
\begin{align}\label{eq:avg-br}
    x_i(t+1) = BR_i(\os_{-i}(t)), & \qquad \forall i \in [n], t \in \bZ_{\ge 0}
\end{align}
Given this, we can write the updated empirical average at time $t+1$ as
\begin{align*}
    \ox_i(t+1) = \ox_i(t) + \frac{1}{t+1} (BR_i(\os_{-i}(t) - \ox_i(t)).
\end{align*}
Notice that $\ox_i(t)$ tracks a BR dynamics with a sequence of decreasing step-sizes that correspond to the harmonic sequence $(\frac{1}{t})_{t \in \bZ_{\ge 1}}$. The harmonic sequence satisfies the following crucial properties: as $t \rightarrow \infty$, the sequence $\frac{1}{t} \rightarrow 0$ but the series $\sum_{k = 1}^t \frac{1}{k} \rightarrow \infty$. These two properties ensure that the dynamics converges for the case $x_{\min} > 0$ using Lemma~\ref{lm:suf}. In fact, we can generalize this dynamics to a weighted average, where the step-size at time $t$ is $\eta_t$, and $\ox_i(t)$ follows
\begin{align}\label{eq:avg-br1}
    \ox_i(t+1) = \ox_i(t) + \eta_t (BR_i(\os_{-i}(t) - \ox_i(t)).
\end{align}
% If $x_{\min} > 0$, as long as the sequence $\eta_t \rightarrow 0$ but the series $\sum_{k = 1}^t \eta_k \rightarrow \infty$ we have convergence.
\begin{theorem}\label{thm:empirical}
If $x_{\min} > 0$, then a dynamics that evolves according to \eqref{eq:avg-br1} converges if the sequence of step-sizes $(\eta_t)_{t \in \bZ_{\ge 1}}$ satisfies: as $t \rightarrow \infty$, the sequence $\eta_t \rightarrow 0$ but the series $\sum_{k = 1}^t \eta_k \rightarrow \infty$.
\end{theorem}
Other examples of step-size sequences that lead to convergence are $\eta_t = 1/t^r$ for $r \in (0, 1]$ and $\eta_t = 1/\log(1+t)$. 
Note that convergence of $\overline{\bx}$ also implies convergence of $\bx$.

\section{Conclusion and Future Research}\label{sec:conclude}
We showed that the continuous BR dynamics, which is motivated by the observation that in certain applications the agents change their actions slowly compared to the feedback they receive from others, converges to the unique equilibrium in Tullock contests with convex costs. We then extended these convergence results to related discrete dynamics with small step sizes. These results indicate that we can expect agents in Tullock contests with convex costs to reach the equilibrium in a decentralized manner.

One open problem is to show convergence (or non-convergence) of the discrete dynamics in Section~\ref{sec:approx} when $x_{\min} = 0$.
Another direction is to study the case when the agents move at rates different than the continuous BR dynamics. 
If all agents move at a rate that is the same constant factor of the continuous BR dynamics rate, then the trajectory followed by the agents remains the same, and the dynamics converges.
On the other hand, if the relative rates of the agents can change arbitrarily over time, we may have non-convergence.
But there are scenarios between these two extreme cases for which convergence properties are unknown, e.g., when the agents move at rates that are different constant factors of the continuous BR dynamics rates. The potential function used in this paper does not easily extend to this case.
Further details are discussed in Appendix~\ref{app:speed}.
Another unexplored direction is to study learning in Tullock contests when the agents get only probabilistic feedback. In many practical applications, the proportional allocation function corresponds to the probability of allocating an indivisible item (and not the fraction of the item allocated to the agent). Here, the agent may only know his own actions and whether or not he won the indivisible item but may not know the actions of others.

\section*{Acknowledgments}
The authors wish to thank the WINE (2024) reviewers for suggestions that helped improve the presentation of the paper. A.G.\ was supported by EPSRC Grant EP/X040461/1 ``Optimisation for Game Theory and Machine Learning''.

\bibliographystyle{splncs04}
\bibliography{ref}

\appendix
% \section{Proofs from Section~\ref{sec:approx}}\label{app:approx}
\section{Omitted Proofs}\label{app:proofs}
\subsection{Proof of Theorem~\ref{thm:contConverge} (Lower Bound)}
\begin{proof}%[Theorem~\ref{thm:contConverge} (Lower Bound)]
We provide an example where it takes $\Omega(\log\left(\frac{1}{\epsilon}\right) + \log(V(0)))$ time to converge to an $\epsilon$-approximate equilibrium. 

Let there be $n = 2$ homogeneous agents with linear cost function $c_1(y) = c_2(y) = y/4$ for any $y \ge 0$. It can be easily derived that the unique equilibrium is $\bx^* = (1, 1)$ and that there is a closed-form formula for the best-response $BR_i(s_{-i}) = 2\sqrt{s_{-i}} - s_{-i}$ (see, e.g., \cite{vojnovic2015contest}). 

Let $\bx(0) = (y(0), y(0))$, where we assume that $y(0)$ is sufficiently large and far away from the equilibrium value $1$. As the two players are homogeneous and start from the same action, they will maintain the same action $\bx(t) = (y(t), y(t))$, for some $y(t)$, for all $t \ge 0$. We suppress the dependency on $t$ to avoid clutter. Let us track the evolution of $y$. From equation \eqref{eq:dx}, we have
\begin{align*}
    \frac{dy}{dt} = \frac{d x_i}{d t} &= BR_i(s_{-i}) - x_{i} = BR_i(y) - y = (2\sqrt{y} - y) - y = 2(\sqrt{y} - y).
\end{align*}

Let us now compute the potential function $V$. We have
\begin{align*}
    V &= \sum_i (u_i(BR_i(s_{-i}), s_{-i}) - u_i(x_i, s_{-i})) = 2 \left( u_1(2\sqrt{y} - y, y) - u_1(y, y) \right) \\
    &= 2 \left( \frac{2\sqrt{y} - y}{2\sqrt{y} - y + y} - \frac{2\sqrt{y} - y}{4} - \left ( \frac{y}{y + y} - \frac{y}{4} \right) \right) = 1 + y - 2\sqrt{y} = (\sqrt{y} - 1)^2.
\end{align*}
Further, we can find the rate of change of $V$ using the rate of change of $y$ as
\begin{align*}
    &\frac{d V}{d y} = \frac{d}{dy} (\sqrt{y} - 1)^2 = \frac{2 (\sqrt{y} - 1)}{2 \sqrt{y}} = \frac{\sqrt{y} - 1}{\sqrt{y}}, \\
    &\frac{d V}{d t} = \frac{d V}{d y} \frac{d y}{d t} = \frac{\sqrt{y} - 1}{\sqrt{y}} 2 (\sqrt{y} - y)  = -2(\sqrt{y} - 1)^2 = -2V  \implies V(t) = V(0) e^{-2t}.
\end{align*}
\sloppy
By the definition of $V$ and the symmetry of the two agents, at an $\epsilon$-approximate equilibrium, $V(t) = 2\epsilon$. So, it takes exactly $t = \frac{1}{2} \ln(\frac{1}{2\epsilon}) + \frac{1}{2} \ln(V(0)) = \Omega(\log\left(\frac{1}{\epsilon}\right) + \log(V(0)))$ time to reach the $\epsilon$-approximate equilibrium.
\qed \end{proof}
\subsection{Proof of Lemma~\ref{lm:suf}}
\begin{proof}%[Lemma~\ref{lm:suf}]
We shall suppress the dependency on $t$ for conciseness, e.g., write $\bx(t)$ as $\bx$ and $x_i(t)$ as $x_i$. 
Let $y_i = BR_i(s_{-i})$.
Given the minimum output level $x_{\min} \ge 0$, $y_i$ satisfies the first-order conditions given below, which can be derived following steps similar to the derivation of conditions \eqref{eq:br1}~and~\eqref{eq:br2}.
\begin{align}
    \frac{s_{-i}}{(y_i + s_{-i})^2} &= c_i'(y_i), & \text{ if $s_{-i} c_i'(x_{\min}) < 1$,} \label{eq:br3} \\
    y_i &= x_{\min}, & \text{ if $s_{-i} c_i'(x_{\min}) \ge 1$}. \label{eq:br4}
\end{align}
From equation~\eqref{eq:V1}, we know that the potential $V(t) = V(\bx(t))$ is given by
\begin{align*}
    V(t) = \sum_i \frac{y_i}{y_i + s_{-i}} - \sum_i c_i(y_i) + \sum_i c_i(x_i) - 1.
\end{align*}
In Lemma~\ref{lm:V}, we showed that $\frac{dV}{dt} < 0$, i.e., $\lim_{\Delta t \rightarrow 0} \frac{V(t+\Delta t) - V(t)}{\Delta t} < 0 $, which implies $\lim_{\Delta t \downarrow 0} (V(t+\Delta t) - V(t)) < 0 $. We shall now consider small positive values for $\Delta t$, not necessarily limiting to $0$, such that we can still guarantee that $V(t+\Delta t) - V(t) < 0$.

Let $\Delta t = \alpha$ be the step-size. Given the current profile $\bx$ and the best-response profile $\by = (BR_i(s_{-i}))_{i \in [n]}$, the profile after the $\alpha$-step is $\alpha (\by - \bx) + \bx$. We want to show that $V(t+\alpha) = V(\alpha (\by - \bx) + \bx) < V(t) = V(\bx)$.
The multivariate Taylor's expansion of $V(\alpha (\by - \bx) + \bx)$ w.r.t.~$V(\bx)$ with a second-order error term is given by
\begin{align}\label{eq:taylor}
    V(\alpha (\by - \bx) + \bx) &= V(\bx) + \alpha (\by - \bx)^\intercal \nabla V(\bx) + \frac{\alpha^2}{2}(\by - \bx)^\intercal \nabla^2 V(\hat{\bx}) (\by - \bx), 
\end{align}
where $\hat{\bx} = \bx + \beta (\by - \bx)$ for some value $\beta \in [0, \alpha]$, and where $\nabla V(\bz) = (\frac{\partial V(\bz)}{\partial z_i})_{i \in [n]}$ and $\nabla^2 V(\bz) = (\frac{\partial^2 V(\bz)}{\partial z_i \partial z_j})_{i,j \in [n]}$ for any given profile $\bz$.

% , and where $\dot{V}$ and $\ddot{V}$ denote the first and second-order time derivatives of $V$. 
% We aim to show that 
% \begin{align*}
%     V(t+\alpha) - V(t) = \alpha (\dot{V}(t) + \alpha \ddot{V}(t+\beta)) < 0,
% \end{align*}
% for suitable values of $\alpha$.

% As discussed in Lemma~\ref{lm:V}, $\dot{V} = \sum_i \frac{\partial V}{\partial x_i} \frac{d x_i}{dt}$, where $\frac{d x_i}{dt} = y_i - x_i$. Let $\nabla V(x) = (\frac{\partial V}{\partial x_i})_{i \in [n]}$. So, $\dot{V}$ can be written as $\dot{V} = (\by - \bx)^\intercal \nabla V(\bx)$, where $t$ is being suppressed.

% Similarly, for any profile $\bz$, let $\nabla^2 V(\bz) = (\frac{\partial^2 V(\bz)}{\partial z_i \partial z_j})_{i,j \in [n]}$. Notice that $\ddot{V}$ needs to be evaluated at $t+\beta$, i.e., at $\hat{\bx} = \bx + \beta (\by - \bx)$. So, $\ddot{V}(t+\beta)$ is given by
% \begin{align*}
%     \ddot{V}(t+\beta) &= (\hat{\bx} - \bx)^\intercal \nabla^2 V(\hat{\bx}) (\hat{\bx} - \bx) \\
%     &= \beta^2 (\by - \bx)^\intercal \nabla^2 V(\hat{\bx}) (\by - \bx)
% \end{align*}

Let us first compute $\nabla^2 V(\bx)$.\footnote{
As discussed in Lemma~\ref{lm:V}, all points except when $s_{-i} = \frac{1}{c_i'(x_{\min})}$ are continuously differentiable (and can be shown to be twice continuously differentiable by extending the same argument). When $s_{-i} = \frac{1}{c_i'(x_{\min})}$, the derivative of $y_i$ w.r.t.~$s_{-i}$ is not well-defined. However, the derivative is well-defined and bounded at all points in its neighborhood. This allowed us to extend the analysis for the differentiable points to this non-differentiable point. A similar but more detailed analysis can be done for computing $\nabla^2 V(\bx)$ to get the same results, but here let us restrict our focus on the \textit{generic} case of $s_{-i} \neq \frac{1}{c_i'(x_{\min})}$.
} 
As $s_{-i} = \sum_{j \neq i} x_j$, so $\frac{d s_{-i}}{d x_i} = 0$ and $\frac{d s_{-i}}{d x_j} = 1$ for all $j \neq i$.
From equation \eqref{eq:dBR1}, we have
\begin{align*}
    \frac{\partial y_i}{\partial x_j} = \frac{y_i - s_{-i}}{2 s_{-i} + (y_i + s_{-i})^3 c_i''(y_i)},
\end{align*}
for $j \neq i$ and $y_i > x_{\min}$, and $0$ otherwise.
As derived in equation \eqref{eq:dV1}, we know that
\begin{align*}
    \frac{\partial V}{\partial x_i} = c_i'(x_i) - \sum_{k \neq i} \frac{y_k}{(y_k + s_{-k})^2}.
\end{align*}
Let $E = \{ k \in [n] \mid y_k > x_{\min} \}$. 
Differentiating again w.r.t.~$x_j$, if $j = i$, we have
\begin{align*}
    \frac{\partial^2 V}{\partial x_i^2} &= c_i''(x_i) + \sum_{k \in E, k \neq i} \left( \frac{2y_k}{(y_k + s_{-k})^3} + \left( \frac{2y_k}{(y_k + s_{-k})^3}  - \frac{1}{(y_k + s_{-k})^2}  \right) \frac{\partial y_i}{\partial x_k} \right) \\
% \end{align*}
% \begin{align*}
    % \implies \frac{\partial^2 V}{\partial x_i^2} 
    &= c_i''(x_i) + \sum_{k \in E, k \neq i} \frac{2y_k}{(y_k + s_{-k})^3} + \sum_{k \in E, k \neq i} \frac{\partial y_i}{\partial x_k} \frac{y_k - s_{-k}}{(y_k + s_{-k})^3}.
\end{align*}
Plugging in the value of $\frac{\partial y_i}{\partial x_k}$ for $k \in E$, we get
\begin{align*}
    &\frac{\partial^2 V}{\partial x_i^2} = c_i''(x_i) + \sum_{k \in E, k \neq i} \frac{2y_k}{(y_k + s_{-k})^3} + \sum_{k \in E, k \neq i} \frac{(y_k - s_{-k})^2}{(y_k + s_{-k})^3 (2s_{-k} + (y_k + s_{-k})^3 c_k''(y_k) )}. %\\
    % &= c_i''(x_i) + \sum_{k \in E, k \neq i} \frac{(y_k + s_{-k})^2 + 2 y_k (y_k + s_{-k})^3 c_k''(y_k)}{(y_k + s_{-k})^3 (2s_{-k} + (y_k + s_{-k})^3 c_k''(y_k) )}
\end{align*}
Let $\eta_k = (y_k + s_{-k})^3 c_k''(y_k)$. The above equation can be rewritten as
\begin{align*}
    \frac{\partial^2 V}{\partial x_i^2} &= c_i''(x_i)  + \sum_{k \in E, k \neq i} \frac{2y_k}{(y_k + s_{-k})^3} + \sum_{k \in E, k \neq i} \frac{(y_k - s_{-k})^2}{(y_k + s_{-k})^3 (2s_{-k} + \eta_k )} \\
    &= c_i''(x_i) + \sum_{k \in E, k \neq i} \frac{(y_k + s_{-k})^2 + 2 y_k \eta_k}{(y_k + s_{-k})^3 (2s_{-k} + \eta_k )}.
\end{align*}
Let $a_i = c_i''(x_i)$. Let $b_i =  \frac{(y_i + s_{-i})^2 + 2 y_i \eta_i}{(y_i + s_{-i})^3 (2s_{-i} + \eta_i )}$ if $i \in E$ and $b_i = 0$ if $i \notin E$. Notice that $a_i \ge 0$ and $b_i \ge 0$ for all $i \in [n]$. We have 
\begin{align*}
    \frac{\partial^2 V}{\partial x_i^2} = a_i + (\sum_i b) - b_i.
\end{align*}
Following a similar sequence of steps, we can compute $\frac{\partial^2 V}{\partial x_i \partial x_j}$ for $j \neq i$ as
\begin{align*}
    \frac{\partial^2 V}{\partial x_i \partial x_j} = (\sum_i b) - b_i - b_j.
\end{align*}
Let $\ba = (a_i)_{i \in [n]}$ and $\bb = (b_i)_{i \in [n]}$. Let $\be = (1, 1, \ldots, 1)$ denote the $n$-dimensional vector of all $1$s. Let $A$ denote the $n \times n$ diagonal matrix with $A_{ii} = a_i + b_i$ and $A_{ij} = 0$ for $j \neq i$. Putting everything together, we can write $\nabla^2 V(\bx)$ as
\begin{align}
    \nabla^2 V(\bx) = (\sum_i b) \be \be^{\intercal} - \bb \be^\intercal - \be \bb^\intercal + A.
\end{align}
Let us now consider a vector $\bw \in \bR^n$. We can compute $\bw^\intercal \nabla^2 V(\bx) \bw$ as
\begin{align*}
    \bw^\intercal \nabla^2 V(\bx) \bw &= \bw^\intercal ((\sum_i b) \be \be^{\intercal} - \bb \be^\intercal - \be \bb^\intercal + A) \bw \\
    &= (\sum_i b_i) (\sum_i w_i)^2 - (\sum_i b_i w_i) (\sum_i w_i) - (\sum_i w_i) (\sum_i b_i w_i) + \sum_i (a_i + b_i) w_i^2  \\
    &= \sum_i b_i \left( (\sum_j w_j)^2 - 2 w_i (\sum_j w_j) + w_i^2 \right) + \sum_i a_i w_i^2 \\
    &= \sum_i b_i (\sum_{j \neq i} w_j)^2 + \sum_i a_i w_i^2. 
\end{align*}

Now, we know that $b_i =  \frac{(y_i + s_{-i})^2 + 2 y_i \eta_i}{(y_i + s_{-i})^3 (2s_{-i} + \eta_i )}$ for $i \in E$. Let us find a suitable upper bound on $b_i$. Let the function $g : \bR_{\ge 0} \rightarrow \bR_{\ge 0}$ be defined as 
\begin{align*}
    g(z) = \frac{(y_i + s_{-i})^2 + 2 y_i z}{(y_i + s_{-i})^3 (2s_{-i} + z )}.
\end{align*}
Notice that $b_i = g(\eta_i)$.
Differentiating $g$ w.r.t.~$z$, we get
\begin{align*}
    \frac{dg}{dz} &= \frac{2y_i}{(y_i + s_{-i})^3 (2s_{-i} + z )} - \frac{(y_i + s_{-i})^2 + 2 y_i z}{(y_i + s_{-i})^3 (2s_{-i} + z )^2}\\
    &= \frac{2y_i (2s_{-i}+z) - (y_i + s_{-i})^2 - 2 y_i z}{(y_i + s_{-i})^3 (2s_{-i} + z )^2} = \frac{- (y_i - s_{-i})^2}{(y_i + s_{-i})^3 (2s_{-i} + z )^2}.
\end{align*}
So, $\frac{dg}{dz}$ is less than or equal to $0$ for all $y_i \ge 0$, $s_{-i} \ge 0$, and $z \ge 0$, which are satisfied by the definition of $y_i$, $s_{-i}$, and $z$. So, $g(z)$ is maximized at $z = 0$, which implies that
\begin{align*}
    b_i = g(\eta_i) \le \max_{z \ge 0} g(z) = g(0) \le \frac{1}{2(y_i + s_{-i})s_{-i}}.
\end{align*}
Plugging back this bound on $b_i$ and the value of $a_i$ to $\bw^\intercal \nabla^2 V(\bx) \bw$, we get
\begin{align*}
    \bw^\intercal \nabla^2 V(\bx) \bw \le \sum_{i \in E} \frac{(\sum_{j \neq i} w_j)^2}{2(y_i + s_{-i})s_{-i}} + \sum_i c_i''(x_i) w_i^2.
\end{align*}
In the Taylor's expansion in equation~\eqref{eq:taylor}, $\nabla^2 V$ is evaluated at $\hat{\bx}$. Let $\hat{s}_{-i} = \sum_{j \neq i} \hat{x}_j$ and $\hat{y}_i = BR_i(\hat{s}_{-i})$. Setting the vector $\bw = \by - \bx$, we get
\begin{align*}
    &(\by - \bx)^\intercal \nabla^2 V(\hat{\bx}) (\by - \bx) \le \sum_{i \in E} \frac{(\sum_{j \neq i} (y_i - x_i))^2}{2(\hat{y}_i + \hat{s}_{-i})\hat{s}_{-i}} + \sum_i c_i''(\hat{x}_i) (y_i - x_i)^2.
\end{align*}
Notice that $\hat{s}_{-i} = \alpha \sum_{j \neq i} y_j + (1 - \alpha) s_{-i} \ge (1 - \alpha) s_{-i}$. So, $(\hat{y}_i + \hat{s}_{-i})\hat{s}_{-i} \ge (1-\alpha)^2 s_{-i}$. As $\alpha$ is a small constant, $(1-\alpha)^2$ is bounded away from $0$, e.g., if $\alpha \le 1/2$, then $(1-\alpha)^2 \ge 1/4$. Let $\sigma = \sum_i y_i$. For $\alpha \le 1/2$, we get
\begin{align*}
    &(\by - \bx)^\intercal \nabla^2 V(\hat{\bx}) (\by - \bx) \le \sum_i \frac{2 (\sigma - y_i - s_{-i}))^2}{s_{-i}^2} + \sum_i B_2 (y_i - x_i)^2,
\end{align*}
as $c_i''(z) \le B_2$ for any $z$. Let $\alpha$ be equal to the smaller among $1/2$ and
\[
    \alpha \le \frac{\sum_i \frac{y_i (\sigma - y_i - s_{-i})^2}{\sigma (y_i + s_{-i})^2} }{ \frac{B_2}{2} \sum_i (y_i - x_i)^2 + \sum_{i \in E} \frac{(\sigma - y_i - s_{-i})^2}{s_{-i}^2} }.
\]
Notice that $\sum_i \frac{y_i (\sigma - y_i - s_{-i})^2}{\sigma (y_i + s_{-i})^2} \le - ((\by - \bx)^\intercal \nabla V(\bx) + V(\bx))$ from Lemma~\ref{lm:V}. So, we get
\begin{align*}
    &\alpha \le \frac{-2(\by - \bx)^\intercal \nabla V(\bx) - 2V(\bx)}{(\by - \bx)^\intercal \nabla^2 V(\hat{\bx}) (\by - \bx)} \implies (\by - \bx)^\intercal \nabla V(\bx) + \frac{\alpha (\by - \bx)^\intercal \nabla^2 V(\hat{\bx}) (\by - \bx)}{2} \le -V(\bx).
\end{align*}
Plugging this into the Taylor's expansion in equation~\eqref{eq:taylor}, we get
\begin{align*}
    V(\alpha (\by - \bx) + \bx) \le (1 - \alpha) V(\bx),
\end{align*}
as required.
\qed \end{proof}

\subsection{Proof of Theorem~\ref{thm:suf1}}
\begin{proof}%[Proof of Theorem~\ref{thm:suf1}]
Note that for any $\bz \in \bR^n$ and $n \ge 2$, we have $\sum_i z_i^2 \le \sum_i (\sum_{j \neq i} z_j)$ because
\begin{multline*}
    \sum_i (\sum_{j \neq i} z_j)^2 = \sum_i ((\sum_j z_j) - z_i)^2 = n (\sum_j z_j)^2 + \sum_i z_i^2 - 2 (\sum_j z_j) (\sum_i z_i) \\
    = (n - 2) (\sum_j z_j)^2 + \sum_i z_i^2.
\end{multline*}
Using this, we have $\sum_i (y_i - x_i)^2 \le \sum_i (\sigma - y_i - s_{-i})^2$.

As $x_i, y_i \in [x_{\min}, 1]$ for all $i$, we have $\frac{\sigma (y_i + s_{-i})^2}{y_i} \le \frac{n^3}{x_{\min}}$ and $\frac{1}{s_{-i}^2} \le \frac{1}{(n-1)^2 x_{\min}^2}$. So, $H(\bx) \le \frac{B_2 n^3}{2 x_{\min}} + \frac{n^3}{(n-1)^2 x^3_{min}} = O(\frac{(1+B_2) n}{x_{\min}^3})$ assuming $x_{\min} \le \frac{1}{n}$.

For any step-size of $\alpha > 0$ satisfying $\alpha \le \frac{1}{H(\bx)}$ for all $\bx$, from Lemma~\ref{lm:suf}, we know that $V(t+\alpha) \le (1-\alpha) V(t)$. After $k$ steps, we have
\begin{align*}
    V(k\alpha) \le (1 - \alpha)^k V(0) \le e^{-k\alpha} V(0).
\end{align*}
If we take $k = \frac{1}{\alpha} \log(\frac{V(0)}{\epsilon})$, then $V(k\alpha) \le \epsilon$, which implies that we have reached an $\epsilon$-approximate equilibrium, as required.
\qed \end{proof}
\subsection{Proof of Lemma~\ref{lm:nec-n}}
\begin{proof}%[Proof of Lemma~\ref{lm:nec-n}]
We construct a simple example with $n$ homogeneous agents where the BR dynamics cycles. Each agent has a cost function $c_i(x_i) = \frac{n-1}{n^2} x_i$. It is not hard to derive that the unique equilibrium is at $(1, 1, \ldots, 1)$ (see, e.g., \cite{vojnovic2015contest}). We will construct a simple two-step cycle where the agents play $(x, x, \ldots, x)$ and $(y, y, \ldots, y)$ repeatedly, where $x < 1$ and $y > 1$. This implies that if the agents enter this cycle, then they will never reach an $\epsilon$-approximate equilibrium for small enough $\epsilon$.

Let us try to construct this cycle, i.e., we try to compute the values of $x$ and $y$ given $\Delta t$, which will then tell us the required bound on $\Delta t$ to have such a cycle. 
From the first-order conditions, for a linear cost function, we have a closed-form formula for the BR as
\begin{align*}
    &\frac{\partial u_i(z, s_{-i})}{\partial z} = 0 \implies \frac{s_{-i}}{(z + s_{-i})^2} = \frac{n-1}{n^2} \implies BR_i(s_{-i}) = n \sqrt{\frac{s_{-i}}{n-1}} - s_{-i}.
\end{align*}
As $x$ is the next action of each agent if everyone is playing $y$, so
\begin{align*}
    x = y + \Delta t \cdot (BR((n-1)y) - y) = y + \Delta t \cdot (n \sqrt{y} - (n-1)y - y) = y + n \Delta t (\sqrt{y} - y).
\end{align*}
Let $\beta = n \Delta t$, we get 
\[
    x = \beta \sqrt{y} + (1 - \beta) y.
\]
By symmetry, as $y$ is the next action for an agent if everyone is playing $x$, we get
\[
    y = \beta \sqrt{x} + (1 - \beta) x \Longleftrightarrow (y - (1-\beta)x)^2 = \beta^2 x.
\]
Combining the two equations in $x$ and $y$, in particular, replacing $x$ to get an equation in a single variable $y$, we get
\begin{align*}
     (y - (1-\beta)(\beta \sqrt{y} + (1 - \beta) y)))^2 &= \beta^2 (\beta \sqrt{y} + (1 - \beta) y) \\
     \Longleftrightarrow  ( \beta (2 - \beta) y - \beta (1-\beta) \sqrt{y})^2 &= \beta^2 (\beta \sqrt{y} + (1 - \beta) y) \\
     \Longleftrightarrow (2 - \beta)^2 y^2 + (1-\beta)^2 y - 2(2-\beta)(1-\beta) y \sqrt{y} &= \beta \sqrt{y} + (1 - \beta) y \\
     \Longleftrightarrow  \sqrt{y} ((2 - \beta)^2 y \sqrt{y} - 2(2-\beta)(1-\beta) y - \beta (1-\beta) \sqrt{y} - \beta) &= 0.
\end{align*}
From the above equation, we get our first root as $\sqrt{y} = 0 \implies y = 0$. This is the degenerate solution that corresponds to everyone playing $0$, and therefore, the BR also being $0$ (see Section~\ref{sec:prelim} for a discussion on BR to $0$). Let us focus on the other roots. For simplicity, let $z = \sqrt{y}$. We have
\begin{align*}
    (2 - \beta)^2 z^3 - 2(2-\beta)(1-\beta) z^2 - \beta (1-\beta) z - \beta = 0.
\end{align*}
Notice that $z = 1$ is a root because
\begin{align*}
    (2 - \beta)^2 - 2(2-\beta)(1-\beta) - \beta (1-\beta) - \beta &= (2 - \beta) ((2 - \beta) - 2(1-\beta)) - \beta + \beta^2 - \beta\\
    &= (2 - \beta) \beta - 2 \beta + \beta^2 = 0.
\end{align*}
This root $\sqrt{y} = z = 1 \implies y = 1$ corresponds to the equilibrium $(1, 1, \ldots, 1)$. Let us now derive the two other roots, which are of our primary interest as they give us the cycle. First, let us factor out the root $z=1$ to convert the cubic equation to a quadratic equation. Let $\hz = z - 1 \Longleftrightarrow z = \hz + 1$. Replacing $z$ by $\hz + 1$, we get
\begin{align*}
    &(2 - \beta)^2 (\hz + 1)^3 - 2(2-\beta)(1-\beta) (\hz + 1)^2 - \beta (1-\beta) (\hz + 1) - \beta = 0 \\
    &\Longleftrightarrow \hz^3 (2 - \beta)^2 + \hz^2 (3(2 - \beta)^2 - 2(2-\beta)(1-\beta)) \\
    & \qquad\qquad + \hz (3(2 - \beta)^2 - 4(2-\beta)(1-\beta) - \beta(1-\beta)) \\
    & \qquad\qquad + (2 - \beta)^2 - 2(2-\beta)(1-\beta) - \beta (1-\beta) - \beta = 0\\
    &\Longleftrightarrow \hz(\hz^2 (2-\beta)^2 + (2-\beta)(4-\beta) \hz + (4 - \beta)) = 0.
\end{align*}
Solving the above equation, in addition to the root $\hz = 0$ that we have already considered, we get the following two roots
\begin{align*}
    &\hz = \frac{\beta - 4 \pm \sqrt{\beta(\beta - 4)}}{2(2-\beta)} \Longleftrightarrow z = \hz + 1 = \frac{\beta \pm \sqrt{\beta(\beta - 4)}}{2(\beta-2)} \\
    &\qquad\Longleftrightarrow y = z^2 = \frac{\beta (\beta - 2) \pm \beta \sqrt{\beta(\beta - 4)}}{2(\beta-2)^2}.
\end{align*}
In particular, $x = \frac{\beta (\beta - 2) - \beta \sqrt{\beta(\beta - 4)}}{2(\beta-2)^2} < 1$ and $y = \frac{\beta (\beta - 2) - \beta \sqrt{\beta(\beta - 4)}}{2(\beta-2)^2} > 1$ are the two roots, assuming $\beta \ge 4$. In particular, if we take $\beta = 6$, we get $x = \frac{3(2-\sqrt{3})}{8} = 0.1005$ and $y = \frac{3(2+\sqrt{3})}{8} = 1.3995$. As $\beta = \Delta t \cdot n = 6$ implies that $\Delta t = \frac{6}{n}$, we get a two step cycle for $\Delta t = \frac{6}{n}$. We can create such cycles for all $\Delta t > \frac{4}{n}$, so $\Delta t \le \frac{4}{n} = O(\frac{1}{n})$ is necessary for convergence.
\qed \end{proof}
\subsection{Proof of Lemma~\ref{lm:nec-cost}}
\begin{proof}%[Proof of Lemma~\ref{lm:nec-cost}]
We consider examples with two agents.
Let the first agent have a cost function $c_1(z) = z$ and the second agent have a cost function $c_2(z) = z/d$ for all $z \ge 0$ and some $d \ge 1$. 
Notice that agent $2$ has a lower cost for the same output than agent $1$, and $c_1'(z) = 1$ and $c_2'(z) = 1/d$. So, $\frac{\min_{i,z} c_i'(z)}{\max_{i,z} c_i'(z)} = \frac{1}{d}$. Let $\alpha = 1 / \Delta t$. We next show non-convergence for $\Delta t = \Omega(1/d) \Longleftrightarrow \alpha = O(d)$, as required.

As the cost functions of the agents is linear, using the first-order conditions given in equation~\ref{eq:du}, the BR of agent $1$ is given by the explicit formula $BR_1(s_{-1}) = BR_1(x_2) = \sqrt{x_2} - x_2$. Similarly, the BR of agent $2$ is $BR_2(s_{-2}) = BR_2(x_1) = \sqrt{d x_1} - x_1$. The unique equilibrium can also be explicitly computed and is equal to $\left( \frac{d}{(1+d)^2}, \frac{d^2}{(1+d)^2} \right)$.

The non-linearity of the BR dynamics makes it tedious to analytically compute the cycles for non-homogeneous agents, especially when $\Delta t$ is strictly less than $1$ . For example, if $d = 16$ and $\Delta t = 0.5$, starting from the state $(1/10, 1/10)$, the dynamics ends up in the cycle in Table~\ref{tab:cycle} where $\bx(t) = \bx(t+6)$. So, we compute the cycles using numerical simulations.

We simulate the BR dynamics starting from the profile $(1/10, 1/10)$. Given a value of $d$, we find the critical value of $\alpha$, denoted by $\alpha^*(d)$, such that the dynamics converges if $\Delta t < \frac{1}{\alpha^*(d)}$ but goes into a cycle if $\Delta t \ge \frac{1}{\alpha^*(d)}$. We do this by a simple binary search over the values of $\Delta t = 1/\alpha$. We notice that $\alpha^*(d)$ is an exact linear function of $d$, as required; plotted in Figure~\ref{fig:dependencyOnCost}. The code is included in the supplementary material.
\begin{table}[t]
    \centering
    \begin{tabular}{|c|c c|}
    \hline
    & $x_1(t)$ & $x_2(t)$ \\
    \hline
    $t$ & $0.021697$ & $0.923661$ \\
    $t+1$ & $0.029555$ & $0.745583$ \\
    $t+2$ & $0.073722$ & $0.701843$ \\
    $t+3$ & $0.104820$ & $0.857095$ \\
    $t+4$ & $0.086759$ & $1.023655$ \\
    $t+5$ & $0.043385$ & $1.057547$ \\
    $t+6$ & $0.021697$ & $0.923661$ \\
    $t+7$ & $0.029555$ & $0.745583$ \\
    $\ldots$ & $\ldots$ & $\ldots$ \\
    \hline
    \end{tabular}
    \caption{Cycle two non-homogeneous agents ($d = 10$).}
    \label{tab:cycle}
\end{table}
\begin{figure}[t]
    \centering
    \includegraphics[width=0.5\linewidth]{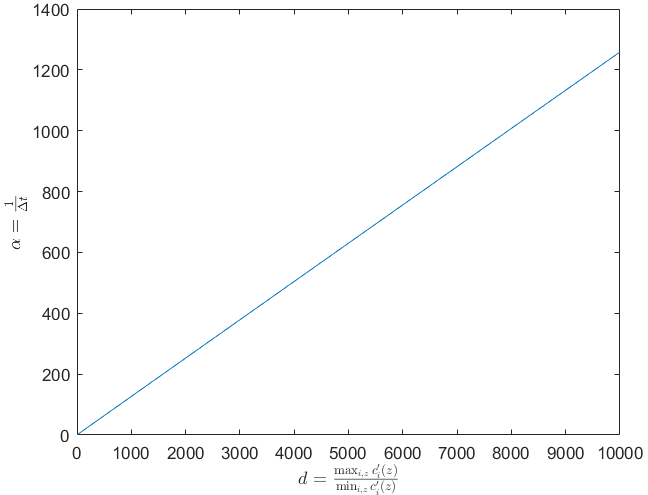}
    \caption{Dependency of step-size on cost ratio.}
    \label{fig:dependencyOnCost}
\end{figure}
\qed \end{proof}
\subsection{Proof of Theorem~\ref{thm:empirical}}
\begin{proof}%[Proof of Theorem~\ref{lm:empirical}]
From Lemma~\ref{lm:suf} and Theorem~\ref{thm:suf1}, we know that at time $t$, if the step-size $\eta_t \le \alpha$, where $\alpha = \Theta\left(\frac{x_{\min}^3}{n(1+B_2)}\right)$, then the potential $V$ decreases by a factor of $(1-\eta_t)$.
As $\eta_t \rightarrow 0$, there exists a $\tau$ such that $\eta_t \le \alpha$ and $V(t+1) \le (1 - \eta_t) V(t)$ for all $t \ge \tau$. So, we get for $t \ge \tau$,
\begin{align*}
    V(t+1) &\le (\prod_{k = \tau}^{t} (1 - \eta_k)) V(\tau) \le (\prod_{k = \tau}^{t} e^{-\eta_k}) V(\tau)  = e^{-\sum_{k = \tau}^{t} \eta_k} V(\tau).
\end{align*}
As $\sum_{k = \tau}^{t} \eta_k \rightarrow \infty$, so $e^{-\sum_{k = \tau}^{t} \eta_k} \rightarrow 0$ and $V(t) \rightarrow 0$.
\qed \end{proof}

\section{Agents Move at Different Rates}\label{app:speed}
When the agents move at arbitrary rates that can depend upon time, it is easy to see that the continuous-time dynamics may not converge because we can simulate a discrete-time dynamics using the continuous-time dynamics by adjusting the rates. In particular, consider the example for two agents presented in Lemma~\ref{lm:nec-cost} where a discrete-time dynamics with a step-size of $1/2$ goes into a cycle. We can simulate this dynamics by slowing down the first agent and allowing the second agent to move, and then slowing down the second agent and allowing the first agent to move, and repeating the process. In a similar manner, the cycles presented in \cite{ghosh2023best1} and \cite{ghosh2023best2}, where only one agent moves at a time, can also be simulated.

On the other hand, simple modifications to the continuous BR dynamics converge. For example, if the dynamics followed by the agents is
\begin{align*}
    \frac{d x_i}{dt} = f(\bx, t) (BR_i(s_{-i}) - x_i),
\end{align*}
where $f(\bx, t)$ is some positive scaling factor that may depend upon the profile $\bx$ and the time $t$ but is the same for all agents. This dynamics follows exactly the same path as the original continuous BR dynamics (albeit at a different speed), and thus converges. Let us consider an additional modification to the dynamics
\begin{align*}
    \frac{d x_i}{dt} = \eta_i f(\bx, t) (BR_i(s_{-i}) - x_i),
\end{align*}
where $\eta_i > 0$ is some positive constant specific to agent $i$. This dynamics speeds up certain agents relative to others. Our intuition is that this dynamics also converges, as we are essentially stretching or squeezing the vector field along different dimensions of the phase space, but a formal analysis remains open. Studying similar variations to the dynamics and showing their convergence (or non-convergence) is an important direction for future work.

\end{document}